\documentclass[journal]{IEEEtran}
\ifCLASSINFOpdf
\else
\fi

\usepackage{balance}

\usepackage{colortbl}
\usepackage{booktabs}

\usepackage{array}
\newcolumntype{P}[1]{>{\centering\arraybackslash}p{#1}}

\usepackage{xcolor}
\definecolor{darkgreen}{rgb}{0.01, 0.75, 0.24}

\usepackage{amssymb}
\usepackage{pifont}

\makeatletter
\newcommand{\setlabel}[1]{\edef\@currentlabel{#1}\label}
\makeatother

\usepackage{amsmath}

\usepackage{hyperref}
\usepackage{cleveref}

\usepackage{amsthm}

\theoremstyle{plain}
\newtheorem{prop}{Proposition}
\crefname{prop}{Proposition}{Propositions}

\crefname{prope}{Property}{Properties}

\newtheorem{definition}{Definition}
\crefname{definition}{Definition}{Definitions}
\newtheorem{lemma}{Lemma}
\crefname{lemma}{Lemma}{Lemmas}
\newtheorem{theorem}{Theorem}
\crefname{theorem}{Theorem}{Theorems}

\crefname{axiom}{Axiom}{Axioms}

\theoremstyle{definition}
\newtheorem{example}{Example}

\crefname{example}{Example}{Examples}

\usepackage{todonotes}

\let\oldtodo\todo
\renewcommand{\todo}[1]{\oldtodo[inline]{#1}}

\newif\ifshowrev
\showrevfalse   
\newenvironment{markrev}{\ifshowrev\color{blue}\fi}{}

\begin{document}
%
\title{A Formal Resilience Framework for\\ Cyber-Physical Embodied Systems under Device-Level Cyberattacks}
%
%
%

\author{Alberto Giaretta
\thanks{A. Giaretta is with the Department of Computer Science, Örebro University, Sweden and the AI, Robotics and Cybersecurity Center (ARC), Örebro University, Sweden e-mail:alberto.giaretta@oru.se.}
\thanks{Manuscript received August 01, 2025; revised August 01, 2025.}}

%
%

\markboth{Journal of \LaTeX\ Class Files,~Vol.~14, No.~8, August~2015}%
{Shell \MakeLowercase{\textit{et al.}}: Bare Demo of IEEEtran.cls for IEEE Journals}
%



\maketitle

\begin{abstract}
In cyber-physical systems (CPSs), fault tolerance is traditionally achieved by analysing sensor and actuator outputs, detecting progressive drift or sudden failures, and initiating suitable tolerance mechanisms. Reasonable under general failure models, this approach fails to capture nuanced disruptions caused by cyberattacks, which may employ subtle strategies. This is particularly critical in embodied CPSs, where computational and physical devices not only have an active role in task completion, but also in embodiment preservation (that is, maintaining the system’s physical integrity). To prevent structural physical damage, embodied CPSs require a framework that enables proactive response to cyberattacks. 
This paper proposes a formal dependability framework that incorporates IDS information into resilience evaluation predicates, enabling assessment of tolerance to disruption and degradation. The framework supports structured reasoning about how cyberattacks affect task execution and embodiment preservation, and whether mitigation strategies must be deployed. Analytical examples demonstrate its analytical capability and soundness, establishing a theoretical foundation for dependable and secure embodied CPSs.
\end{abstract}


%
\IEEEpeerreviewmaketitle

\section{Introduction}\label{sec:intro}
A cyber-physical system (CPS) is a system in which software, computational devices, and physical components (e.g., sensors and actuators), are tightly coupled to achieve a set of goals~\cite{mandrioli2025}. CPSs encompass a wide range of domains, including vehicles, manufacturing plants, medical devices, process controls, power generation and distribution, water management systems, and distributed robotics~\cite{lee2015}. 

Fault tolerance for CPSs revolves around a core concept: obtaining a system that is capable of detecting, diagnosing, isolating, and recovering from failures~\cite{piardi2020}. Traditional strategies for detecting faults are based on the fundamental assumption that faults represent physical anomalies in a component, hence the faults can be observed as abnormal output~\cite{nauta2025,taheri2024,alwan2022}. In principle, fault detection methods can be used to detect some classes of cyberattacks, the ones that evidently affect devices' availability or performance~\cite{park2017}. An example of such category is denial of service (DoS) attacks. 

However, many cyberattacks differ from faults, in that they are orchestrated by an adversary that aims to remain undetected~\cite{taheri2024}, making traditional fault detection ineffective. Instances of these attacks are replay attacks, or timing anomaly attacks. One of the most prominent examples of a stealthy cyberattack on CPSs is Stuxnet, regarded as a landmark in modern cyberwarfare. In its latest version, Stuxnet targets one of the weak points of Iranian uranium-enriching centrifuges, their resistance to rotor wall pressure. By manipulating rotor speeds, for just as little as 30 minutes once a month, Stuxnet was capable of slowly damaging a number of centrifuges. It later became apparent that the malware would have been capable of catastrophic damage, but was specifically designed to create slow, hard-to-detect disruption that undermined Iranian engineers' confidence~\cite{langner2013kill}. Stuxnet highlights how fault detection alone is not sufficient for CPSs.  

Recent advances in AI and robotics have given rise to a new category of CPSs: \textit{embodied CPSs}. These systems are designed to be fully autonomous, capable of updating their own objectives, and interacting with surrounding dynamic environments, including human users~\cite{shea-blymyer2021}. What distinguishes embodied CPSs from general CPSs is that their computational and physical components are not just tightly coupled, but functionally interdependent. Consequently, the system’s ability to interpret its environment and act meaningfully depends on the integrity of its own physical structure. 

In traditional CPSs, cyberattacks may compromise availability or functionality, but embodiment is not inherently at risk. In embodied CPSs, a compromised component may not just disable a function; it may distort perception or control in ways that endanger both task completion and physical integrity. The embodiment is not just a means of execution, but must be protected as a system-critical asset.

We contend that embodied CPSs impose higher requirements for cybersecurity resilience, compared to traditional CPSs. From a dependability standpoint, embodied CPSs must maintain both availability and structural integrity even in the presence of compromised devices. From a cybersecurity standpoint, these systems must integrate intrusion-detection information into their resilience reasoning process. Addressing this interplay, the present work bridges \textit{dependability modelling} and \textit{cybersecurity analytics} within a unified formal framework for resilience reasoning.

On the grounds of this argument, cybersecurity must be treated not as an external aspect to be evaluated separately, but as an intrinsic component of systems' runtime reasoning. Embodied CPSs must continuously assess how cyberattacks affect not only task execution, but the preservation of their own embodiment. To support this, in this paper we propose a formal and theoretical framework that integrates functional goals, embodiment constraints, device-level roles, and intrusion detection system (IDS) signals into a unified model for resilience evaluation. Our framework extends classical fault-tolerance reasoning to include IDS-informed dependability evaluation, thereby bridging traditional reliability analysis with modern intrusion-aware resilience modelling.

\subsection{Contributions}
This paper makes the following key contributions:
\begin{itemize}
  \item We formally characterize the resilience requirements of embodied CPSs under cyberattacks.
  \item We introduce a set of resilience predicates that capture disruption tolerability, degradation thresholds, and mitigation feasibility based on device-level criticality.
  \item We propose a formal theoretical framework that, albeit being IDS-implementation agnostic, integrates IDS information into runtime reasoning to assess and maintain both functional and embodiment integrity.
  \item We prove system-level soundness guarantees through a series of lemmas and resilience theorems, enabling structured, provable resilience evaluation.
\end{itemize}

\section{Related Work}
\label{sec:rel_work}
Scholars have extensively discussed the importance of IDS for CPSs.  
For example, in 2014 Han et al.~\cite{han2014} surveyed the process of integrating intrusion detection mechanisms in CPSs and outlined a list of requirements for making IDSs useful in the CPS context. The authors argue that IDSs for CPSs should be anomaly-based, distributed, and online (i.e., operating as close to real time as possible). In addition, the IDS should be fault-tolerant and avoid exposing any privacy-sensitive data during its analysis.

More recently, Segovia-Ferreira et al.~\cite{segovia2024} surveyed advances in cyber-resilience approaches for CPSs. One important point highlighted in this paper is that two main categories of detection approaches can be deployed. The first is a traditional data-based approach, in which data from devices is analysed through pattern recognition or machine learning (ML) techniques. Although such approaches are capable of detecting cyberattacks, they are not suitable for every kind of cyber-physical attack, as they do not consider the underlying CPS control models. The second category, model-based detection, takes into account the physical model of the system and compares its expected output to the actual output; if the system behaviour is inconsistent with the model's expected output, a cyberattack is assumed to be ongoing. Thanks to this feedback control, this category is more suitable for identifying cyber-physical attacks, but it requires a reliable model. Their observation lays the foundation for this paper. We share the view of Segovia-Ferreira et al.~\cite{segovia2024} that information about cyberattacks does not directly translate into meaningful insight for CPSs regarding their cyber-physical resilience. Our formal framework tackles this exact problem, bridging the gap between cyberattack detection and reasoning about cyber-physical resilience.

Since then, various works have investigated different implementations of IDSs for CPSs, and we present a non-exhaustive list of notable examples. It is worth clarifying that our work does not address the design or implementation of IDSs for embodied CPSs. Instead, we provide a theoretical framework that assumes the availability of IDS outputs and integrates them into a holistic resilience model for embodied CPSs. Although our framework is IDS-agnostic, it is relevant to show that such IDSs have been proposed in the literature.

For example, Freitas de Araujo-Filho et al.~\cite{freitas2021} proposed FID-GAN, an unsupervised IDS that leverages a generative adversarial network (GAN) and a fog-based architecture. Their edge-compatible design, combined with an encoder that accelerates the computation of the reconstruction loss, allows FID-GAN to meet the low-latency requirements outlined by Han et al.~\cite{han2014}. Their detection rates are also improved with respect to baseline approaches such as MAD-GAN~\cite{li2019} and ALAD~\cite{zenati2018}. Pinto et al.~\cite{pinto2022} discuss how modern technology such as the Internet of things (IoT) and edge/fog computing are reshaping design and development processes of cyber-physical production systems (CPPSs). Devices are independently developed and then integrated in the final system, making CPPSs' (and, we argue, CPSs') security different from classic information and communication technology (ICT) security. Therefore, they propose a bio-inspired anomaly-based host IDS (A-HIDS), based on the incremental dendritic cell algorithm (iDCA) technique, showing its suitability for edge devices. A third example of an attack detection approach for CPSs was given by Yan et al.~\cite{yan2019}, who 
propose an ML-based attack detection scheme, which models and monitors the system's behaviour at the physical layer. Their method, based on extreme learning machines (ELMs), uses a hybrid feature set that includes statistical descriptors, physics-based residuals, and deep learned representations from sensor data. While the focus is on detecting attacks at the measurement level, the work supports our assumption that well-calibrated physical-layer IDSs can serve as reliable inputs for higher-level reasoning frameworks. Other works have proposed alternative solutions, and we refer the reader to dedicated surveys on the topic~\cite{han2014,segovia2024,olowononi2021,giraldo2017,mitchell2014}.

Kayan et al.~\cite{kayan2022} note that the existing literature has often focused heavily on IDSs in isolation, neglecting architectural and policy-level vulnerabilities. Our work supports their critique, reinforcing the message that IDSs by themselves are not enough for CPSs, as they provide only a cybersecurity lens on the system's status. But IDSs should not be discarded either. We position IDSs as one element within a resilience framework tailored to embodied CPSs, where the cyber-physical coupling necessitates tight integration of intrusion awareness into goal-critical reasoning.

In the field of control theory for CPSs, various approaches adopting formal methods have been proposed, for modelling both safety and cybersecurity~\cite{yin2020}. Lanotte et al.~\cite{lanotte2017} formalise cyber-physical attacks using a hybrid process calculus (CCPSA) that models how sensor and actuator tampering affects system behaviour over time. However, they assess tolerance based on whether an attack leads to observable deviations, compared to the nominal system. In fact, their modelled IDS monitors the operational temperature of a physical component. In contrast, our approach incorporates cybersecurity IDS' outputs, enabling proactive adjustments before observable degradation occurs. This allows us to respond to latent or low-impact threats, such as dormant malware or early-stage intrusions, that may not yet affect system outputs.

Mouelhi et al.~\cite{mouelhi2019} propose a formal methodology for modelling and verifying CPS resilience that echoes the spirit of our proactive approach. Using UPPAAL timed automata, their framework includes an abstract attack detection mechanism, conceptually similar to a binary IDS output, which sets an intrusion flag upon observing malicious activity in the communication layer. This flag is used to formally drive transitions into degraded operational modes and is included in the verification of safety and liveness properties. This approach aligns closely with our integration of IDS-based attack detection signals into cyberattack resilience modelling for embodied CPSs.

\section{Propositions} \label{sec:informal_propositions}
In this section, we provide two informal propositions that frame the goals we aim to achieve with our formal framework. For each proposition we provide a relevant use case example, highlighting the critical properties that can be achieved through such a formal framework.

\begin{prop}\label{prop_1}
    A robotic framework should be aware of which disruptions and how much performance degradation it can tolerate, based on component criticality and acceptability thresholds.
\end{prop}

\begin{example}
\label{ex:1}
Assume a wheeled robot, equipped with eight cameras and four robotic arms, designed for navigating the environment, picking up objects, and transporting them to different destinations. These components are off-the-shelf IoT devices, developed separately by different manufacturers and integrated into a single embodiment using middleware — for example, the Robot Operating System 2 (ROS2). One camera and one robotic arm are labelled as critical devices, as they represent the minimum set required for the robot to achieve its tasks and preserve embodiment integrity.

Suppose a cyberattack strikes the system and deactivates two out of eight cameras, as well as one robotic arm; however, none of the disrupted devices are labelled as critical. Intuitively, even though the robot has been compromised, it can continue to operate. Although it could perform its tasks more efficiently with all devices functioning, it is still capable of navigating and detecting objects using the critical camera and the remaining functional ones. Once the object has been detected, the robot can use the critical arm to pick it up, potentially assisted by the other operational arms.

In light of these observations, although some of the embodied devices have been disrupted and rendered unusable, the system as a whole remains capable of safe operation. Throughout this paper, we refer to this condition as \emph{tolerable disruption}: a scenario in which the embodiment has been subjected to cyberattacks affecting some devices, but none of the critical ones. A formal definition is introduced in \Cref{sec:formal_def}.
\end{example}

\begin{prop}\label{prop_2}
    A robotic framework should assess whether disruption and degradation can be mitigated by evaluating whether appropriate mitigation actions exist for the compromised devices involved, and applying them.
\end{prop}
\begin{example}\label{ex:2}
Consider the same wheeled robot, previously introduced in \Cref{ex:1}, equipped with eight cameras and four robotic arms, designed for navigating the environment, picking up objects, and transporting them to different destinations.

Suppose a cyberattack compromises the robot and disables the camera marked as critical, along with two non-critical cameras and one non-critical robotic arm. In this case, a required device is rendered unusable, making the disruption not tolerable ($\delta(S) = 0$); even if other devices remain functional, the system cannot proceed safely in this state.

However, the robot still has five functioning cameras, and the framework determines that the perception pipeline can be reconfigured: the critical role can be reassigned to one of the functioning cameras, thereby restoring the system's ability to navigate and detect objects. After this mitigation action, tolerable disruption and degradation are re-established, and the robot is capable of fully resuming task execution.

Throughout this paper, we refer to the ability of robotic systems to evaluate and apply such mitigation strategies to restore tolerability as \emph{mitigation feasibility}. We formalize this concept in \Cref{sec:formal_def}, using the predicate~$\mu(S)$.
\end{example}

\section{Threat Model and Assumptions}
We consider embodied CPSs whose devices may be targeted by adversaries through various cyberattacks aimed at disrupting sensing and actuating mechanisms. Our proposed framework is agnostic with respect to specific attack types and models their effects abstractly through a probabilistic intrusion detection function, $I(d)$ (\Cref{def:prob_ids}). Cyberattacks are therefore represented as adversarial manipulations of device-level operational integrity. 

The adversary is assumed to have partial, but not full, system knowledge, sufficient to compromise on-board devices but not to predict subsequent mitigation strategies. The objective of the proposed framework is to reason formally about system-level resilience and mitigation feasibility, independently of any particular IDS implementation or attack vector.

\section{Formal Definitions}\label{sec:formal_def}
Although the propositions introduced in \Cref{sec:informal_propositions} are conceptually intuitive, a formal framework is required to enable precise and unambiguous evaluation. This section introduces the necessary mappings and predicates that instantiate \Cref{prop_1,prop_2}.

\subsection{Core Definitions}
We consider a modular robotic system \( R \), composed of a finite set of devices \( D = \{d_1, d_2, \dots, d_n\} \). The robot can perform a finite set of tasks \( T = \{t_1, t_2, \dots, t_p\} \), and must preserve a finite set of embodiment-related goals \( G = \{g_1, g_2, \dots, g_q\} \). Each device \( d \in D \) contributes to task execution and goal preservation through criticality mappings. In the following definitions, we use the standard notation \( 2^D \) to denote the power set of \( D \), i.e., the set of all subsets of devices. In addition, we assume all functions and predicates are evaluated at the current time. For notational simplicity, we omit explicit time-dependence throughout the remainder of the paper.

The following assumptions hold throughout this section:
\begin{itemize}
    \item The sets \( D \), \( T \), and \( G \) are fixed during evaluation;
    \item All mappings are assumed to be correct and complete;
    \item A task \( t \in T \) is feasible if at least one non-compromised device supports it with criticality \( \geq 1 \);
    \item A goal \( g \in G \) is feasible under the same condition;
    \item Mitigations, when applied, take effect before the next control cycle.
\end{itemize}

We begin by introducing two foundational mappings that express the criticality of each device in relation to both task execution and embodiment preservation.
\begin{definition}[Task-criticality mapping \(\tau\)]\label{def:task-crit-amp}
We define the task-criticality mapping as a function:
\[
\tau : D \times T \rightarrow C,
\]
where \( C = \{0, 1, 2\} \) denotes the criticality levels: \( 0 \) for none, \( 1 \) for important, and \( 2 \) for required.
\end{definition}

\begin{definition}[Goal-criticality mapping \(\varepsilon\)]\label{def:goal-crit-amp}
We define the goal-criticality mapping as a function:
\[
\varepsilon : D \times G \rightarrow C,
\]
where \( C = \{0, 1, 2\} \) denotes the previously introduced criticality levels.
\end{definition}

Now we define another foundational predicate that models the behaviour of a probabilistic IDS, followed by a predicate that evaluates this probability against a confidence threshold.
\begin{definition}[Probabilistic IDS function]\label{def:prob_ids}
Let $D$ be the set of devices present in the embodied CPS. We define the IDS output as a function
\[
I : D \rightarrow [0,1],
\]
where $I(d)$ denotes the confidence (i.e., estimated probability) that the device $d \in D$ is currently compromised by a cyberattack.
\end{definition}

\begin{definition}[Device-specific IDS threshold function]\label{def:device_threshold}
Let $\kappa_{\mathrm{base}}, \kappa_{\mathrm{crit}} \in [0,1]$ be fixed thresholds such that $\kappa_{\mathrm{crit}} < \kappa_{\mathrm{base}}$.

We define the device-specific threshold function $\kappa : D \rightarrow [0,1]$ as:
\[
\kappa(d) = 
\begin{cases}
\kappa_{\mathrm{crit}}, & \text{if } \exists t \in T_{\mathrm{active}} : \tau(d, t) = 2 \\
                       & \text{or } \exists g \in G_{\mathrm{active}} : \varepsilon(d, g) = 2 \\
\kappa_{\mathrm{base}}, & \text{otherwise}
\end{cases}
\quad ,
\]
where \( T_{\mathrm{active}} \subseteq T \) and \( G_{\mathrm{active}} \subseteq G \) are the currently relevant tasks and embodiment-preserving goals, respectively.
\end{definition}

In this last core definition, we build on \Cref{def:prob_ids,def:device_threshold} to introduce an unambiguous definition of compromised devices. Throughout the paper, when we refer to a set of compromised devices $S$, we refer to the following. 
\begin{definition}[Criticality-aware compromised set]\label{def:ca_compromised}
Given the IDS function $I(d)$ and threshold function $\kappa(d)$, we define the current set of compromised devices $S$ as:
\[
S = \{ d \in D \mid I(d) \geq \kappa(d) \}.
\]
\end{definition}

\subsection{Extended Definitions}
So far, we have provided the core definitions that lay the foundations for our formal framework. Building on these definitions, we now formalise \Cref{prop_1}, which addresses the robot's ability to assess whether a disruption compromises structural or performance integrity.

\begin{definition}[Tolerable Disruption \(\delta\)]\label{def:tol_dis}
We define the tolerable disruption predicate as a function:
\[
\delta : 2^D \rightarrow \{0,1\},
\]
which evaluates whether a set of compromised devices \( S \subseteq D \) includes any device that is critical (i.e., strictly required) for task execution or goal preservation.

Let \( T_{\mathrm{active}} \subseteq T \) and \( G_{\mathrm{active}} \subseteq G \) be the currently relevant tasks and embodiment-preserving goals, as previously defined in \Cref{def:device_threshold}.

The predicate is defined as:
\[
\delta(S) = 1 \Longleftrightarrow
\begin{aligned}
&\ \forall t \in T_{\mathrm{active}},\;
\nexists d \in S:\; \tau(d, t) = 2 \\
&\land\; \forall g \in G_{\mathrm{active}},\;
\nexists d \in S:\; \varepsilon(d, g) = 2
\end{aligned}
\ \ \,.
\]

This condition is satisfied only if no compromised device is critical for any active task or goal.
\end{definition}

With \Cref{def:tol_dis}, we capture the binary nature of disruption, meaning that a device is either functioning or not: tasks are feasible and embodiment is preserved, only if no required device has been compromised. However, this alone does not allow to reflect on the degradation that occurs under partial loss of functionality. We need the concept of tolerable degradation.

To model tolerable degradation, the second component of \Cref{prop_1}, we introduce a continuous performance score and define acceptable thresholds that reflect whether a degraded state remains operationally viable.
\begin{definition}[Tolerable Degradation \(\gamma\)]\label{def:tol_deg}
Let \( \psi : 2^D \rightarrow [0,1] \) be a monotonic degradation function that evaluates system performance under a set of compromised devices \( S \subseteq D \), where 1 denotes full performance and 0 denotes total failure. We assume that \( \psi \) is monotonic non-increasing, reflecting the fact that additional disruption cannot lead to improved performance. Let \( \theta_{\mathrm{crit}}, \theta_{\mathrm{base}} \in [0,1] \) be performance thresholds, with \( \theta_{\mathrm{crit}} > \theta_{\mathrm{base}} \).

We define the tolerable degradation predicate as a function:
\[
\gamma : 2^D \rightarrow \{0,1\},
\]
which evaluates whether a system under disruption remains within acceptable performance bounds.

The predicate is defined as:
\[
\gamma(S) = 1 \Longleftrightarrow \psi(S) \geq \theta(S),
\]
where the threshold \( \theta(S) \) depends on the presence of critical devices, following:
\[
\theta(S) =
\begin{cases}
\theta_{\mathrm{crit}}, & \text{if } \exists d \in S,\, \exists t \in T_{\mathrm{active}}:\; \tau(d, t) = 2 \\
                       & \text{or } \exists d \in S,\, \exists g \in G_{\mathrm{active}}:\; \varepsilon(d, g) = 2 \\
\theta_{\mathrm{base}}, & \text{otherwise}
\end{cases}
\quad .
\]

The components of this condition are interpreted as follows:
\begin{itemize}
    \item \( \gamma(S) = 1 \): performance within acceptable bounds;
    \item \( \theta_{\mathrm{crit}} \): strict threshold if critical devices are affected;
    \item \( \theta_{\mathrm{base}} \): relaxed threshold for non-critical device disruptions.
\end{itemize}
\end{definition}

\noindent This definition captures the second component of \Cref{prop_1}, ensuring that the robot evaluates not only the feasibility of its functions, but also their operational quality under performance loss. It is worth noting that mitigations, when applied, reduce the set of compromised devices before evaluating \( \psi \). Therefore, they do not violate the monotonic non-increasing assumption for \( \psi \).

We now formalise \Cref{prop_2}, which addresses the robot's ability to determine whether a disruption can be sufficiently mitigated, either by isolating affected components, or by deploying appropriate countermeasures. Such mitigation is realized through a policy that assigns a symbolic mitigation action (isolation, reconfiguration, role reassignment, or a combination) to the devices it can address. \begin{markrev} We write \( D_m \subseteq D \) for the set of devices for which a mitigation action exists. The following predicate verifies the existence of a subset of compromised devices whose mitigation restores tolerability.
\end{markrev}

\begin{definition}[Mitigation Feasibility \(\mu\)]\label{def:mit_feas}
Let \( A \) be the set of symbolic mitigating actions, such as isolating compromised devices, reconfiguring system logic, or reassigning criticality to functionally equivalent, non-compromised devices. 

We define the mitigation feasibility predicate as a function:
\[
\mu : 2^D \rightarrow \{0,1\},
\]
which evaluates whether a compromised set \( S \subseteq D \) can be mitigated, at least partially, to restore sufficient system operability. Specifically:
\[
\mu(S) = 1 \Longleftrightarrow \exists M \subseteq S \begin{markrev}\cap D_m: \end{markrev}\;
\delta(S \setminus M) = 1 \land \gamma(S \setminus M) = 1.
\]

This predicate returns true if there exists a subset of compromised devices $M$ for which an appropriate mitigation strategy can be applied, such that both disruption and degradation become tolerable.
\end{definition}

\noindent The predicate \( \mu \) abstracts away from the specific implementation of the mitigation policy. It only verifies the existence of a subset of devices whose mitigation is sufficient to re-establish tolerability.

\subsection{Correctness Lemmas and Resilience Theorems}
We now establish a set of formal lemmas that validate the correctness of the tolerance and mitigation predicates introduced earlier. These results demonstrate that the framework provides consistent guarantees, regarding disruption impact and system resilience.

We begin by formalising the notion of resilience with respect to task execution. The following lemma characterizes when the robot can tolerate a cyberattack without violating its functional goals.

\begin{lemma}[Soundness of $\delta$]\label{lemma:soundness_delta}
If $\delta(S) = 1$, then no critical device for any task $t \in T_{\mathrm{active}}$ or goal $g \in G_{\mathrm{active}}$ is included in $S$.
\end{lemma}

\begin{proof}[Proof sketch]
From the definition of $\delta(S)$, we have:
\[
\delta(S) = 1 \Longleftrightarrow
\begin{aligned}
&\ \forall t \in T_{\mathrm{active}},\;
\nexists d \in S:\; \tau(d, t) = 2 \\
&\land\; \forall g \in G_{\mathrm{active}},\;
\nexists d \in S:\; \varepsilon(d, g) = 2
\end{aligned}
\ \ \,.
\]

This means that for each task $t$ and goal $g$, no device $d$ such that $\tau(d, t) = 2$ or $\varepsilon(d, g) = 2$ is included in $S$.

Therefore, all critical devices for tasks and goals must lie outside of $S$, hence they are operational.
\end{proof}

In parallel, we define resilience in terms of embodiment integrity. The following lemma captures the conditions under which the robot maintains structural coherence despite component-level compromise.
\begin{lemma}[Strictness of $\delta$]\label{lemma:strictness_delta}
If there exists a device \( d \in S \) and a task \( t \in T_{\mathrm{active}} \) such that \( \tau(d, t) = 2 \), or a device \( d \in S \) and a goal \( g \in G_{\mathrm{active}} \) such that \( \varepsilon(d, g) = 2 \), then \( \delta(S) = 0 \).
\end{lemma}

\begin{proof}[Proof sketch]
By definition, \( \delta(S) = 1 \) holds if and only if no critical device for any active task or goal is included in the compromised set \( S \). Therefore, if there exists a device \( d \in S \) that is critical for some task \( t \in T_{\mathrm{active}} \) or goal \( g \in G_{\mathrm{active}} \), the condition for \( \delta(S) = 1 \) is violated. Hence, \( \delta(S) = 0 \).
\end{proof}

\begin{markrev}
Beyond evaluating $\gamma$ at a fixed compromised set, we establish how $\gamma$ behaves as the compromised set changes. The following lemma shows that tolerability is preserved under removal of compromised devices, and dually, that intolerability persists under addition of compromised devices.

\begin{lemma}[Monotonicity of $\gamma$]
\label{lemma:gamma-monotone}
If $S \subseteq S' \subseteq D$ and $\gamma(S') = 1$, then $\gamma(S) = 1$.
\end{lemma}
 
\begin{proof}[Proof sketch]
Since $\psi$ is monotonic non-increasing and $S \subseteq S'$, we have
$\psi(S) \geq \psi(S')$.
The threshold $\theta$ is non-decreasing under set inclusion: by
\Cref{def:tol_deg}, $\theta(S) = \theta_{\mathrm{crit}}$ if $S$ contains a
device with criticality score $2$. Otherwise,  $\theta(S) = \theta_{\mathrm{base}}$, with
$\theta_{\mathrm{crit}} > \theta_{\mathrm{base}}$. Because $S \subseteq S'$, the presence of a critical device in $S$ implies its presence in $S'$, and therefore $\theta(S) \leq \theta(S')$. From $\gamma(S') = 1$ and \Cref{def:tol_deg} we have $\psi(S') \geq \theta(S')$. Chaining these inequalities results in:
\[
\psi(S) \;\geq\; \psi(S') \;\geq\; \theta(S') \;\geq\; \theta(S),
\]
hence $\psi(S) \geq \theta(S)$, so $\gamma(S) = 1$.
\end{proof}
\end{markrev}

\subsection{Theorems}
We now generalize these conditions into a system-level guarantee. \begin{markrev}
Here, resilient denotes that the system either currently tolerates the disruption or possesses a mitigation that restores tolerability; it is a property of the system together with its available recovery capabilities, not of the instantaneous device state alone.\end{markrev} The following Theorem formalises this notion: 
\begin{markrev} the system is resilient exactly when the disruption is either directly tolerable or restorable through mitigation.\end{markrev}

\begin{theorem}[System Resilience Guarantee]\label{theorem:res_guar}
Given a compromised set \( S \subseteq D \), the system is resilient \begin{markrev}if and only if:
\[
\bigl(\delta(S) = 1 \land \gamma(S) = 1\bigr) \lor \mu(S) = 1.
\]\end{markrev}
\end{theorem}
\begin{markrev}
\noindent  Intuitively, the system is resilient if it is either already tolerable or can be made tolerable by mitigating some compromised devices. Note that $\mu(S)=1$ holds whenever $\delta(S)=1$ and $\gamma(S)=1$, since the empty set $M=\emptyset$ trivially satisfies \Cref{def:mit_feas}. We specify the first disjunct to make the no-action case explicit, although $\mu(S)=1$ alone already expresses resilience.
\end{markrev}

\begin{proof}[Proof sketch]
\begin{markrev}
This theorem is a direct application of \Cref{lemma:soundness_delta,lemma:strictness_delta} together with \Cref{def:tol_deg,def:mit_feas}.
\end{markrev}
If \( \delta(S) = 1 \) and \( \gamma(S) = 1 \), then the system maintains both structural integrity and acceptable performance, despite the disruption. 
If this condition does not hold, \begin{markrev}
but \( \mu(S) = 1 \), then by \Cref{def:mit_feas} a mitigation restores both disruption and degradation tolerability. 
\end{markrev}
In either case, the system remains resilient under the definition.
\end{proof}

So far, we have assumed the robot is aware of which devices are compromised. In practice, this information is not self-evident, and it must come from an external source. For example, an IDS. The following theorem shows how system-level resilience can still be preserved under the assumption of sound detection.

\begin{theorem}[IDS-informed resilience evaluation]\label{thm:ids_resilience}
Let $S$ be the current set of compromised devices based on the IDS output $I(d)$ (\Cref{def:prob_ids}) and the threshold function $\kappa(d)$ (\Cref{def:device_threshold}), as defined in \Cref{def:ca_compromised}.

\begin{markrev}While \Cref{theorem:res_guar} determines whether the system is resilient, the following cases distinguish a fully functional system from a degraded one, and identify when such a system is beyond recovery.\end{markrev}

The system's resilience posture can then be evaluated according to the following cases:
\begin{itemize}
  \item If both $\delta(S) = 1$ and $\gamma(S) = 1$, then the system tolerates the disruption without degradation.
  \item \begin{markrev} Otherwise, \end{markrev} if $\mu(S) = 1$, then the system is not disruption tolerant, but there exists a mitigation that restores tolerability, albeit with some possible performance degradation.
  \item \begin{markrev} Otherwise, if $\mu(S) = 0$, \end{markrev} then the disruption is not tolerable, and the system must initiate safe-state mitigation or halt.
\end{itemize}
\end{theorem}

\begin{proof}[Proof sketch]
Let $S$ denote the IDS-reported set of compromised devices. Then we have three possible cases:

\textbf{Case 1:} $\delta(S) = 1$ and $\gamma(S) = 1$.  
By \Cref{def:tol_dis,def:tol_deg}, the system is resilient without further action.

\textbf{Case 2:} $\mu(S) = 1$ \begin{markrev} and Case 1 does not hold.  
By \Cref{def:mit_feas}, there exists $M \subseteq S \cap D_m$ with $\delta(S \setminus M) = 1$ and $\gamma(S \setminus M) = 1$, so applying the corresponding mitigation restores tolerability.
\end{markrev}

\textbf{Case 3:} $\mu(S) = 0$.
By definition, there is no such mitigation policy that can be applied to recover a resilient condition as in Case 2. Therefore, the robot is considered irreparably compromised.
\end{proof}

\begin{table*}[t]
\centering
\caption{\Cref{ex:1}: Task- and goal-criticality mappings for the wheeled robot.}
\label{tab:ex1}
\renewcommand\arraystretch{1.1}
\begin{tabular}{llcccc}
\toprule
\textsc{Device} & \textsc{Role} &
$\tau(d,\text{TransportCube})$ &
$\varepsilon(d,\text{AvoidCollision})$ &
\textsc{Compromised} &
$\kappa(d)$\\
\midrule
\rowcolor{green!5}
\textbf{Camera~$C_1$} & Perception (detection, navigation) & \textbf{2} & \textbf{2} & No & $\kappa_{\mathrm{crit}}=0.4$\\
\rowcolor{red!5}
Cameras~$C_3$,$C_5$ & Perception (redundant) & 1 & 1 & \textbf{Yes} & $\kappa_{\mathrm{base}}=0.8$\\
\rowcolor{green!5}
Cameras~$C_2$,$C_4$,$C_6$--$C_8$ & Perception (redundant) & 1 & 1 & No & $\kappa_{\mathrm{base}}=0.8$\\
\arrayrulecolor{gray!50}\cmidrule(lr){1-6} 
\arrayrulecolor{black}
\rowcolor{green!5}
\textbf{Arm~$A_1$} & Manipulation (pick, place) & \textbf{2} & 1 & No & $\kappa_{\mathrm{crit}}=0.4$\\
\rowcolor{green!5}
Arms~$A_2$--$A_3$ & Manipulation (redundant) & 1 & 1 & No & $\kappa_{\mathrm{base}}=0.8$\\
\rowcolor{red!5}
Arm~$A_4$ & Manipulation (redundant) & 1 & 1 & \textbf{Yes} & $\kappa_{\mathrm{base}}=0.8$\\
\midrule
\multicolumn{6}{p{62em}}{\footnotesize
$\tau,\varepsilon\in\{0,1,2\}$ (0~=~none, 1~=~important, 2~=~required). Rows in light red indicate compromised devices; in light green, non-compromised ones.}\\
\bottomrule
\end{tabular}
\end{table*}

\section{Evaluation}
In the previous section, we have fleshed out the formal components sitting at the foundation of our framework. In this section, we provide a formal, analytical evaluation of such components, first of a disruption tolerant embodiment (\Cref{ex:1}, and then of a non-tolerant one (\Cref{ex:2}).

\subsection{Evaluation of \Cref{ex:1}}
In light of \Cref{def:task-crit-amp,def:goal-crit-amp} for task- and goal-criticality mappings we must enrich \Cref{ex:1}, defining at least one task and one embodiment preservation goal. 

As a brief reminder, the robot is equipped with eight cameras and four robotic arms. The robot's task is the following: to move from point A to point B, pick up a blue-coloured cube, and bring it back to point A. We call this task \textit{TransportCube}, and we mark one camera and one robotic arm as critical devices for completing it. As an embodiment preservation goal, the robot must avoid any type of collision to preserve its embodiment integrity and we name this goal \textit{AvoidCollision}. For ensuring this goal, only one camera is marked as critical. 

Following \Cref{def:prob_ids,def:device_threshold}, we define for each device the IDS thresholds that determine when a device is considered compromised. The IDS sensitivity depends on the device’s criticality: it must react more readily to anomalies on critical devices. Consequently, the threshold $\kappa_{\mathrm{crit}}$ is set lower than the baseline $\kappa_{\mathrm{base}}$ used for non-critical devices. \Cref{tab:ex1} summarises the task- and goal-criticality mappings together with their respective IDS thresholds.

We now prove that the robot in \Cref{ex:1} is tolerant to disruptions. For brevity, we have defined only one task and goal, \textit{TransportCube} and \textit{AvoidCollision}, respectively. Therefore, \( T_{\mathrm{active}} = (TransportCube) \) and \( G_{\mathrm{active}} = (AvoidCollision)\). The embodiment is under attack and two cameras are deactivated, $C_3$ and $C_5$, as well as the robotic arm $A_4$, entailing that the set $S$ of compromised devices is composed as $S = (C_3,C_5,A_4)$. 

Only the cameras $C_3$ and $C_5$ are involved in goal \textit{AvoidCollision}, but all three devices in $S$ are involved in task \textit{TransportCube}'s completion. Following \Cref{def:tol_dis}, we must therefore verify that for none of these devices neither the task-criticality mapping $\tau(d,\text{TransportCube}) = 2$, nor the goal-criticality mapping $\varepsilon(d,\text{AvoidCollision}) = 2$. From \Cref{tab:ex1}, we can verify that both task- and goal-criticality mappings for each of the compromised devices in $S$ is 1. Consequently, $\delta(S) = 1$, and the embodiment is disruption tolerant. 
\begin{markrev}
Since no compromised device in $S$ is critical, the applicable threshold is $\theta(S) = \theta_{\mathrm{base}}$; assuming $\psi(S) \geq \theta_{\mathrm{base}}$ (e.g., $\psi(S) = 0.95 \geq 0.75$), we obtain $\gamma(S) = 1$.
\end{markrev}
Therefore, according to \Cref{theorem:res_guar}, the system is resilient and, following \Cref{thm:ids_resilience}, it tolerates the disruption without any observable degradation.

\subsection{Evaluation of \Cref{ex:2}}
\begin{table*}[t]
\centering
\caption{\Cref{ex:2}: Task- and goal-criticality mappings for the wheeled robot.}
\label{tab:ex2}
\renewcommand\arraystretch{1.1}
\begin{tabular}{llcccc}
\toprule
\textsc{Device} & \textsc{Role} &
$\tau(d,\text{TransportCube})$ &
$\varepsilon(d,\text{AvoidCollision})$ &
\textsc{Compromised} &
$\kappa(d)$\\
\midrule
\rowcolor{red!5}
\textbf{Camera~$C_1$} & Perception (detection, navigation) & \textbf{2} & \textbf{2} & \textbf{Yes} & 
$\kappa_{\mathrm{crit}}=0.4$\\
\rowcolor{red!5}
Cameras~$C_3$,$C_5$ & Perception (redundant) & 1 & 1 & \textbf{Yes} & $\kappa_{\mathrm{base}}=0.8$\\
\rowcolor{green!5}
Cameras~$C_2$,$C_4$,$C_6$--$C_8$ & Perception (redundant)  & 1 & 1 & No & $\kappa_{\mathrm{base}}=0.8$\\
\arrayrulecolor{gray!50}\cmidrule(lr){1-6} 
\arrayrulecolor{black}
\rowcolor{green!5}
\textbf{Arm~$A_1$} & Manipulation (pick, place) & \textbf{2} & 1 & No & $\kappa_{\mathrm{crit}}=0.4$\\
\rowcolor{green!5}
Arms~$A_2$--$A_3$ & Manipulation (redundant)  & 1 & 1 & No & $\kappa_{\mathrm{base}}=0.8$\\
\rowcolor{red!5}
Arm~$A_4$ & Manipulation (redundant) & 1 & 1 & \textbf{Yes} & $\kappa_{\mathrm{base}}=0.8$\\
\midrule
\multicolumn{6}{p{62em}}{\footnotesize
$\tau,\varepsilon\in\{0,1,2\}$ (0~=~none, 1~=~important, 2~=~required). Rows in light red indicate compromised devices; in light green, non-compromised ones.}\\
\bottomrule
\end{tabular}
\end{table*}

We now turn to \Cref{ex:2}. The task and goal remain identical to those defined in \Cref{ex:1}, namely \textit{TransportCube} and \textit{AvoidCollision}. The set of compromised devices is also unchanged, except for one addition: camera~$C_1$, which is now compromised. This modification is summarised in \Cref{tab:ex2}.

The evaluation proceeds analogously to \Cref{ex:1}, with the key difference being the inclusion of a critical device among the compromised devices. Specifically, camera~$C_1$ is marked as critical for both task and embodiment preservation, expressed as $\tau(C_1,\text{TransportCube})=2$ and $\varepsilon(C_1,\text{AvoidCollision})=2$. Consequently, according to \Cref{def:tol_dis} and \Cref{lemma:strictness_delta}, $\delta(S)=0$ and the embodiment is not disruption-tolerant. We must therefore assess whether the embodiment satisfies, at minimum, the condition of degradation tolerance as expressed in \Cref{def:tol_deg}.

Assume that system performance under disruption is $\psi(S) = 0.9$, and that the performance thresholds are defined as follows: $\theta_{\mathrm{crit}} = 0.85 > \theta_{\mathrm{base}} = 0.75$. The set of compromised devices $S$ includes at least one critical device, camera $C_1$, with $\tau(C_1, TransportCube) = 2$ and $\varepsilon(C_1, AvoidCollision) = 2$. Following \Cref{def:tol_deg}, the applicable threshold is $\theta(S) = \theta_{\mathrm{crit}} = 0.85$ and $\psi(S)= 0.9 \geq \theta(S) = 0.85$, thus $\gamma(S) = 1$.

Although quantitative performance remains within acceptable bounds,
\begin{markrev}
the first resilience condition of \Cref{theorem:res_guar} does not hold, since $\delta(S)=0$. Resilience therefore depends on the second condition, $\mu(S)=1$, i.e., whether a mitigation exists which restores tolerability.
\end{markrev}

We assume that the system is equipped with a reconfiguration strategy $m_{reconfig}$, capable of isolating compromised devices and re-assigning their tasks to other, similar devices. Therefore, $m_{reconfig}$, can isolate a compromised device $M \subseteq S$ and assign its tasks to alternative on-board devices. In our case example, $M=\{C_1\}$ indicates the system's ability to execute $m_{reconfig}$, substituting the critical camera $C_1$ with another on-board camera $C_n$, selected from a subset of safe devices $\{C_2,C_4,C_6,..,C_8\}$. 

\begin{markrev}
Substituting $C_1$ with $C_2$ yields the residual set $S \setminus M = \{C_3, C_5, A_4\}$, which contains no critical device; hence $\delta(S \setminus M) = 1$. Since $S \setminus M \subseteq S$ and $\gamma(S) = 1$ was established above, \Cref{lemma:gamma-monotone} gives $\gamma(S \setminus M) = 1$. Following \Cref{def:mit_feas}, $\mu(S) = 1$, and the system has restored both disruption and degradation tolerability.
\end{markrev}
In summary, the evaluations of \Cref{ex:1} and \Cref{ex:2} jointly demonstrate the framework’s internal consistency and reasoning power. When no critical components are compromised, the predicates correctly yield disruption and degradation tolerability ($\delta(S) = \gamma(S) = 1$), confirming inherent resilience. When a critical component is affected, 
\begin{markrev} 
the first resilience condition fails ($\delta(S) = 0$), and resilience is determined by mitigation feasibility: the framework verifies $\mu(S) = 1$, establishing that tolerability is restorable.
\end{markrev}

These examples collectively illustrate that the proposed formalism enables precise, verifiable reasoning about resilience, degradation, and mitigation in embodied CPSs under device-level cyberattacks.

\section{Future Work}
While the proposed framework provides a theoretical foundation for reasoning about resilience in embodied CPSs, further experiments and evaluations are warranted.

First, we will investigate how the framework can be deployed on real embodied CPSs and to what extent such integration can effectively mitigate cyberattacks. The current framework does not model probabilistic false positives or false negatives of IDS outputs, nor have we validated its runtime scalability. Future work will address these aspects through simulation and empirical studies (e.g., ROS~2 with Gazebo and real robotic platforms), evaluating the predicates $\delta$, $\gamma$, and $\mu$, under live or simulated IDS outputs.

Second, we will examine how different IDS sensitivity settings and mitigation strategies influence overall resilience outcomes, with the aim of identifying optimal trade-offs between sensitivity, false alarms, and operational continuity.

Finally, we intend to explore alignment between the proposed formalism and emerging assurance and certification standards for safety-critical CPSs, thereby enabling its integration into structured dependability and resilience arguments.

\section{Conclusion}
This paper presented a formal framework for evaluating resilience in embodied CPSs under device-level disruptions. The proposed framework unifies classical fault-tolerance reasoning with IDS-based dependability evaluation, advancing formal assurance methods for resilient CPSs, and providing the foundation for future empirical validation and integration with real-world robotic systems.

By introducing the predicates of tolerable disruption ($\delta$), tolerable degradation ($\gamma$), and mitigation feasibility ($\mu$), the framework enables structured reasoning about whether an embodied system can continue operating safely and effectively in the presence of compromised components. Formal lemmas and theorems demonstrated the internal soundness of the model, while analytical examples showed its ability to distinguish between disruption-tolerant and non-tolerant embodiments and to validate feasible mitigation strategies. 

\balance

\section*{Acknowledgment}
This work has been supported by the Wallenberg AI, Autonomous Systems and Software Program (WASP) funded by the Knut and Alice Wallenberg Foundation.

The author would like to thank Prof. A. Saffiotti, Dr. M. Kashani, and E. Miotto for their insights and comments that helped to further refine this paper. 

The author acknowledges the use of GPT-5 (OpenAI, San Francisco, CA, USA) and Claude Opus 4.8 (Anthropic, San Francisco, CA, USA) to assist in drafting and refining portions of the formal reasoning and proofs. All AI-generated outputs were extensively and critically reviewed, corrected, and verified by the author for technical accuracy and logical consistency. The author takes full responsibility for the validity of all results presented in this paper.

\bibliographystyle{IEEEtran}
\bibliography{bibliography}

\end{document}